\documentclass[conference]{IEEEtran}

\usepackage{amsmath,amssymb,amsthm,mathtools}       
\newtheorem{theorem}{Theorem}[section]

\newtheorem{proposition}[theorem]{Proposition}
\newtheorem{corollary}[theorem]{Corollary}

\ifCLASSINFOpdf
\else
\fi
\begin{document}
%
\title{FISGuard: Defending Against Membership Inference via Fixed Input Subspaces}

\author{
	\IEEEauthorblockN{Haocheng Jiang}
	\IEEEauthorblockA{
		Hubei University of Technology \\
		jianghaocheng@hbut.edu.cn
	}
	\and
	\IEEEauthorblockN{Hua Shen}
	\IEEEauthorblockA{
		Hubei University of Technology \\
		nancy78733@126.com
	}
}
	

%


\maketitle

\begin{abstract}
As large language models are increasingly adopted in federated learning, protecting user privacy while performing parameter-efficient fine-tuning on distributed private data has become an important challenge. Although clients only share gradients instead of directly uploading raw data, the shared gradients may still leak membership information about training samples. ProjRes (S\&P, 2026) further increases this risk: with less information and without accessing model outputs, an attacker can effectively distinguish members from non-members solely based on the projection residual between a candidate representation and the subspace induced by server-observable gradients. Existing defenses against membership inference mostly rely on gradient perturbation or regularization, which can not only degrade model utility but also fail to effectively defend against the membership inference attack introduced by ProjRes, which exploits the geometric structure of gradients. 

To address this issue, we propose FISGuard, a lightweight defense. Its key idea is to construct and fix a low-dimensional representation subspace using independent public data, thereby restricting the space through which private representations are exposed via gradients while preserving the primary information required for downstream tasks. This substantially reduces the projection-residual discrepancy between members and non-members. 

We evaluate FISGuard against five representative defense methods across three NLP datasets, two LLMs, and two fine-tuning strategies, Adapter and LoRA. The results show that FISGuard reduces the ProjRes attack AUC to near the random-guessing level of 0.5 in most settings, while maintaining downstream task performance close to that of the undefended model and introducing only limited computational overhead, thereby achieving a favorable privacy--utility trade-off.
\end{abstract}



%
\IEEEpeerreviewmaketitle

\section{Introduction}
Federated large language models enable multiple clients to collaboratively fine-tune a pretrained model without sharing their local data \cite{HLL2025,KQL2024}. To reduce the training and communication costs on clients, existing systems commonly adopt parameter-efficient fine-tuning methods, such as Adapters and LoRA, which update only a small number of additional parameters \cite{DCL2026}. However, despite their compact size and structured form, the gradients of these parameters may still reveal information about clients’ training data\cite{DLZ2025}.
Membership inference attacks aim to determine whether a candidate record was used to train a target client’s model. Membership itself may constitute sensitive information, as it can reveal whether a user appears in a medical, financial, or other privacy-sensitive dataset \cite{SSS2017}.

ProjRes performs membership inference against federated large language models by exploiting the gradient structure of trainable linear layers. Its key observation is that the gradient of a linear layer is jointly determined by local sample representations and backpropagated error signals. Consequently, the representations of training samples contribute to the construction of the gradient row space. ProjRes \cite{DCL2026} extracts this space from server-visible gradients and infers membership based on the projection residual of a candidate representation with respect to the gradient space. Because member samples directly participate in gradient computation, their representations can typically be explained more accurately by the gradient space. In contrast, non-member representations are more likely to retain larger components outside that space.

This attack is not limited to Adapters. Any fully connected mapping that directly receives private representations, is trained by the client, and exposes its gradients to the server may create a similar input-side gradient attack surface. The Down projection in a standard Adapter, the input-side low-rank factor in LoRA, and ordinary fully connected layers involved in private fine-tuning all exhibit this basic structure. The original ProjRes attack achieves strong membership discrimination across all evaluated model–dataset combinations. Nevertheless, existing research still lacks a defense specifically designed around the general linear-gradient structure exploited by ProjRes.

Existing defenses against membership inference primarily include output protection, training regularization, gradient perturbation, and secure aggregation \cite{HSS2022}. Methods such as output perturbation, confidence modification, and knowledge distillation mainly restrict an attacker’s ability to extract membership information from prediction probabilities or loss values. ProjRes, however, does not rely on final model outputs and instead directly analyzes client-uploaded gradients. Even when a model produces similar predictions for members and non-members, the structural relationship exploited by ProjRes may persist as long as private representations continue to participate in the construction of the visible gradient space.

Training regularization \cite{MIH2021} and early stopping \cite{PL2002} generally reduce membership leakage by mitigating overfitting. However, the basis of ProjRes is not entirely attributable to model memorization. Once a private sample participates in even a single gradient computation, its representation may become embedded in the resulting gradient structure. Therefore, reducing overfitting alone cannot fundamentally prevent private representations from forming a gradient space observable to the server.

Gradient clipping \cite{CWH2020} and differential privacy \cite{DC2006} provide more general privacy protection, but they operate primarily by perturbing client updates. To disrupt the gradient directions and subspace structure exploited by ProjRes, the added noise must substantially alter the geometry of the gradients, which may further intensify the privacy–utility trade-off in federated large language models. Secure aggregation \cite{BIK2016} changes the attacker’s observation model by hiding individual client updates. However, it relies on the additional protocol assumption that the server can observe only the aggregate of multiple client updates, and it does not alter the leakage structure inherent in an individual client’s gradient. Overall, existing defenses either fail to directly address the gradient structure observed by ProjRes or require random perturbation or additional system assumptions.

These limitations motivate a different question: instead of concealing information after gradients have already been generated, can we redesign the parameterization of linear layers before private training begins, such that the system no longer produces the gradient structure required by ProjRes?

This question follows directly from the mechanism of ProjRes. ProjRes can distinguish members from non-members because a client uses a limited number of local samples to form an incomplete gradient subspace within a high-dimensional representation space. Member representations participate in constructing this subspace and are therefore more likely to lie within it, whereas non-member representations are more likely to occupy uncovered directions. Accordingly, the effectiveness of ProjRes depends on two conditions. First, the original private representations must directly enter server-visible gradients. Second, the gradient space must cover only part of the candidate representation space, thereby producing different projection residuals for members and non-members.

An effective defense must therefore address both conditions. First, it should prevent the original high-dimensional private representations from directly generating server-visible input-side gradients. Second, after removing this direct attack surface, it should constrain the remaining gradients to a sufficiently low-dimensional and fixed space that can be fully covered by the gradient span. In this setting, both member and non-member representations can be fully represented within the remaining gradient space, leaving the original ProjRes attack without projection residuals that can be used for membership ranking.

Based on this insight, we propose a fixed input-subspace reparameterization framework for trainable linear mappings. A task-specific linear update is expressed as $\Delta W = UQ,$
where $Q$ is an input-side projection determined before private training begins and permanently frozen thereafter, while $U$ is the only output-side mapping trained by the client. The original representation is first transformed by $Q$ into a fixed low-dimensional coordinate space, and private client data are used only to learn how these coordinates should be mapped to the output space.

The framework reconstructs the attack surface of ProjRes through three consecutive steps. First, it freezes the input-side projection, preventing the original high-dimensional private representations from directly producing observable input-side weight gradients. Second, it moves the remaining server-visible gradients from the original high-dimensional representation space into a fixed low-dimensional space. Third, it ensures that this low-dimensional space can be fully covered by the remaining gradient span, thereby eliminating the out-of-space residuals used by the original ProjRes attack to distinguish members from non-members.

This design is not a local modification specific to a particular Adapter parameter. Rather, it is a general principle applicable to trainable fully connected mappings. For a bottleneck Adapter, the framework can be instantiated as a Fixed-Down Adapter, in which the Down projection is frozen and only the Up projection is trained. For LoRA, the input-side low-rank factor can be fixed while only the output-side factor is trained. For a general fully connected layer, the gradient space can be reconstructed in the same way by introducing a low-rank residual branch with a fixed input projection. We use the Fixed-Down Adapter and LoRA as the primary implementations and experimental vehicles, but the theoretical formulation and scope of the method are not limited to these two architectures.

To preserve useful representation information within the fixed subspace, we construct the input projection using an independent public dataset rather than an arbitrary random matrix. The public data determine which input directions the model is allowed to use, while the private data are used only to learn how these directions should be combined for the downstream task. This separation ensures that the input-side geometry is determined before private training begins and does not adapt to client samples. We also employ function-preserving initialization so that the newly introduced low-rank branch does not alter the original forward behavior of the pretrained model at the start of training.

Compared with existing defenses, the proposed framework offers three advantages. First, it directly targets the gradient attack surface exploited by ProjRes, rather than modifying model outputs or attack scores after the leakage signal has already formed. Second, it does not rely on random gradient noise, thereby avoiding the need to select a noise magnitude and the corresponding utility degradation caused by perturbation. Third, it changes only the organization of trainable parameters and does not require modifications to the client-side task loss, the server aggregation rule, or the model inference interface. In the Adapter instantiation considered in this work, public data are used to construct a fixed Down projection, while only the Up projection and task head are trained, making the method compatible with existing federated training pipelines.

The main contributions of this work are as follows:

\begin{itemize}
	\item  We provide a unified analysis of ProjRes from the perspective of general trainable fully connected layers and show that its attack capability originates from private input representations directly participating in the construction of a server-observable gradient space. Based on this analysis, we identify two key conditions required for the original ProjRes attack to achieve membership discrimination: direct gradient exposure of the original private representations and incomplete coverage of the candidate representation space by the gradient space.
	\item We propose a fixed input-subspace reparameterization framework that decomposes a linear update into a fixed input projection and a trainable output mapping. The framework first removes the original high-dimensional input-side gradient and then constrains the remaining gradient to a fixed low-dimensional space that can be fully covered. We further instantiate the framework as a Fixed-Down Adapter and provide a unified formulation for extending it to LoRA and general fully connected residual updates.
	\item  Using the Fixed-Down Adapter as a representative implementation, we conduct experiments on BERT, GPT-2, and multiple text classification tasks to evaluate the proposed framework in terms of membership privacy protection, task performance, and applicability across different models and datasets.
\end{itemize}

\section{Preliminary}
\subsection{Federated Fine-tuning of Large Language Models}
Federated fine-tuning of large language models (FedLLM fine-tuning) refers to the application of federated learning to the downstream adaptation of pre-trained large language models (LLMs), enabling multiple clients to collaboratively fine-tune a shared model without directly exchanging their local raw data \cite{LM2020,HCB2021}. A central server first initializes a pre-trained language model and distributes it to the participating clients. Each client independently fine-tunes the received model using its local private dataset and uploads only the updated model parameters or parameter updates to the server. The server then aggregates the updates from participating clients to obtain a new global model, which is subsequently redistributed for the next communication round. This process iteratively follows the paradigm of \emph{global model distribution--local fine-tuning--parameter uploading--server aggregation} until the global model converges.

Formally, consider a federated learning system consisting of $K$ clients. The $k$-th client owns a local dataset $\mathcal{D}_k$ containing $n_k$ samples, and the total number of samples across all clients is
$
	N = \sum_{k=1}^{K} n_k.
$
The objective of federated LLM fine-tuning is to optimize a global model without centralizing the local datasets, which can be formulated as
\[
	\min_{\Theta} F(\Theta)
	=
	\sum_{k=1}^{K}
	\frac{n_k}{N}
	F_k(\Theta),
\]
where $\Theta$ denotes the model parameters involved in fine-tuning, and $F_k(\Theta)$ denotes the local training objective of client $k$. Specifically,
\[
	F_k(\Theta)
	=
	\frac{1}{n_k}
	\sum_{(x_i,y_i)\in\mathcal{D}_k}
	\ell\left(f_{\Theta}(x_i),y_i\right),
\]
where $f_{\Theta}$ denotes the model parameterized by $\Theta$, and $\ell(\cdot)$ denotes the task-specific loss function.

At communication round $t$, the server distributes the current global model $\Theta^{t}$ to the participating clients. Each client initializes its local model with $\Theta^{t}$ and performs local fine-tuning on $\mathcal{D}_k$, yielding the locally updated model $\Theta_k^{t+1}$. Under the Federated Averaging (FedAvg) framework, the server aggregates the local models according to their local data sizes:
\[
	\Theta^{t+1}
	=
	\sum_{k=1}^{K}
	\frac{n_k}{N}
	\Theta_k^{t+1}.
\]
The resulting $\Theta^{t+1}$ is then used as the global model for the next communication round.

For modern LLMs, directly fine-tuning and communicating all model parameters can incur substantial computational and communication overhead. Therefore, parameter-efficient fine-tuning (PEFT) is commonly integrated into federated LLM fine-tuning. In this setting, the model parameters can be decomposed as
$\Theta = \{\theta, \phi\},$
where $\theta$ denotes the frozen parameters of the pre-trained backbone, and $\phi$ denotes a small set of trainable parameters introduced or selected for downstream adaptation. Accordingly, the federated optimization objective becomes
\[
	\min_{\phi}
	F(\phi;\theta)
	=
	\sum_{k=1}^{K}
	\frac{n_k}{N}
	F_k(\phi;\theta),
\]
where $\theta$ remains fixed throughout the fine-tuning process and only $\phi$ is locally optimized.

At round $t$, client $k$ receives the global PEFT parameters $\phi^t$ and performs local optimization on its private dataset:
$
	\phi_k^{t+1}
	=
	\operatorname{LocalTrain}
	\left(\theta,\phi^t,\mathcal{D}_k\right).
$
The locally updated PEFT parameters are then uploaded to the server and aggregated as
$
	\phi^{t+1}
	=
	\sum_{k=1}^{K}
	\frac{n_k}{N}
	\phi_k^{t+1}.
$
Therefore, instead of repeatedly communicating the entire LLM, parameter-efficient federated fine-tuning only requires the trainable parameters $\phi$ to participate in local optimization and federated aggregation.

In this work, we consider two representative PEFT mechanisms, namely Adapter and Low-Rank Adaptation (LoRA). For Adapter-based fine-tuning, a lightweight bottleneck module is inserted into the pre-trained Transformer while the backbone parameters remain frozen. Given a hidden representation $\mathbf{h}\in\mathbb{R}^{d}$, an Adapter can be expressed as
$
	\operatorname{Adapter}(\mathbf{h})
	=
	\mathbf{h}
	+
	W_{\mathrm{up}}
	\sigma\left(
	W_{\mathrm{down}}\mathbf{h}
	\right),
$
where
$
	W_{\mathrm{down}}\in\mathbb{R}^{r\times d},
	W_{\mathrm{up}}\in\mathbb{R}^{d\times r},
	r \ll d,
$
and $\sigma(\cdot)$ denotes a nonlinear activation function. In this case, the trainable parameter set $\phi$ mainly consists of the parameters of the inserted Adapter modules.
For LoRA-based fine-tuning, the pre-trained weight matrix $W_0\in\mathbb{R}^{d_{\mathrm{out}}\times d_{\mathrm{in}}}$ is kept frozen, while its task-specific update is parameterized by two low-rank matrices:
$
	W
	=
	W_0 + \Delta W
	=
	W_0 + BA,
$
where
$
	A\in\mathbb{R}^{r\times d_{\mathrm{in}}},
	B\in\mathbb{R}^{d_{\mathrm{out}}\times r},
	r \ll \min(d_{\mathrm{in}},d_{\mathrm{out}}).
$

Thus, only the low-rank matrices $A$ and $B$ are optimized during local fine-tuning and communicated during federated aggregation. Under this unified formulation, $\phi$ represents the trainable Adapter parameters in Adapter-based fine-tuning or the low-rank parameters in LoRA-based fine-tuning.

\section{Related Work}
Existing defenses against membership inference attacks mainly reduce privacy leakage by mitigating model overfitting, introducing differential privacy, and perturbing membership-related information. These methods have also been gradually extended to federated learning settings. However, attack signals have expanded from prediction outputs to gradients, training trajectories, and internal representations. Existing defenses therefore still face trade-offs among privacy protection, model utility, and system overhead. In particular, they have difficulty addressing gradient leakage caused by parameter-efficient fine-tuning in federated large language models.  

One category of membership inference defenses reduces the behavioral gap between members and non-members by mitigating model overfitting and memorization of training samples. Ying et al.\cite{YZL2020} combine L2 regularization with differential privacy. Their method reduces membership leakage by suppressing overfitting and perturbing model outputs. Nasr et al.\cite{NSH2018} further propose Adversarial Regularization. They formulate training of the target model and the membership inference attack model as a min-max adversarial game. This allows the target model to maintain classification performance while learning output distributions that make members and non-members harder to distinguish. Li and Ribeiro\cite{LLR2021} propose MMD+Mix-up. It uses Mix-up to reduce the generalization gap and applies MMD regularization to constrain the prediction distributions of the training and validation sets. For federated learning, Ahmed et al.\cite{ASH2025} propose MemberShield. It combines soft labels with early stopping to reduce overconfidence and sample memorization in client models. These methods are relatively easy to implement and require only limited changes to the model architecture. However, their effectiveness largely depends on the correlation between membership leakage and model overfitting. They also lack rigorous privacy guarantees. As a result, exploitable membership information may still be exposed under white-box, trajectory-based, or adaptive membership inference attacks.  

Another category of methods reduces membership leakage by weakening the direct association between sensitive samples and the final model, or by limiting the model's memorization of specific training samples. Hu et al.\cite{HLL2022} propose DMIG. It trains a generative model on sensitive data and then uses synthetic data to train the final classifier. This prevents the target model from directly learning real training samples. Jarin and Eshete\cite{JE2022} propose MIAShield, which partitions the training data across multiple submodels. During inference, it actively excludes models that have been trained on the queried sample. This weakens membership signals caused by sample memorization. Bai et al.\cite{BLZ2026} further propose CoFedMID for multi-round trajectory-based MIAs in federated learning. It combines class-guided data partitioning, utility-aware sample compensation, and aggregation-cancelable perturbation. These mechanisms jointly reduce the persistent exposure of sensitive samples during training across multiple clients. Such methods can directly weaken the relationship between the model and sensitive samples. However, they often require additional generative models, model ensembles, sample management, or client coordination mechanisms. These requirements increase training, storage, and communication costs, while also making system implementation more complex.  

Differential privacy is a widely used formal privacy-preserving mechanism in federated learning. Naseri et al.\cite{NHD2020} systematically evaluate local differential privacy (LDP) and central differential privacy (CDP). Their results show that gradient clipping and noise injection can reduce the success rates of membership inference and some backdoor attacks. CDP generally provides a better privacy-utility trade-off. For different federated learning settings, Hu et al.\cite{HWL2024} propose PDP-FLames. It exploits the sparsity of knowledge graph gradients to privately select and perturb active gradients. Zhao et al.\cite{ZLS2025} propose ADPF, which combines particle swarm optimization with a Stackelberg game. It dynamically adjusts the privacy budget according to client heterogeneity and attack behavior. To reduce the performance degradation caused by gradient clipping and random noise in DP-FL, Shi et al.\cite{SWS2025} propose DP-FedSAM and DP-FedSAM-topk. These methods use Sharpness-Aware Minimization and update sparsification to improve robustness to DP noise. Yuan et al.\cite{YCW2025} further propose DP-FedPUAC. It uses adaptive clipping thresholds and dynamic local iterations to improve the privacy-utility trade-off and communication efficiency. From a robustness perspective, Qi et al.\cite{QWH2024} point out that DP noise may interfere with the detection of malicious updates. They therefore propose Robust-DPFL, which identifies poisoning updates based on differences in gradient element distributions. Although some of these studies do not specifically target membership inference, they share the goal of improving the practicality and robustness of DP-FL. Nevertheless, differential privacy still involves a significant privacy-utility trade-off. Stronger privacy protection usually requires stricter clipping and larger amounts of random noise. This can degrade model accuracy, convergence speed, and training stability.  

In addition to directly applying differential privacy, some studies reduce the information exposed during training or apply targeted perturbations to potential membership signals. Abourayya et al.\cite{AKR2025} propose FEDCT. Clients share only hard labels for public unlabeled data, without exchanging model parameters or soft confidence scores. Its extension, DP-FEDCT, further randomizes the discrete labels. This reduces membership leakage while lowering communication overhead. Yang et al.\cite{YYH2023} propose Client-level Input Perturbation (CIP). It learns personalized input perturbations for different clients and keeps them locally. The goal is to make the model behavior of members and non-members more similar. At the model-parameter level, Shen et al.\cite{SMX2024} propose MemDefense. It identifies parameters that have little effect on classification performance but carry strong membership signals. Pruning-based perturbation is then applied to weaken white-box MIAs. Another work, Accuracy-Lossless Model Perturbation\cite{YFF2022}, allows the server to perturb the distributed model using secret random parameters. The perturbations are exactly canceled after aggregation. This hides the true model information while preserving model updates that are nearly identical to those of the original FedAvg. Compared with generic random noise, these methods can suppress membership information in a more targeted manner and often achieve a better privacy-utility balance. However, they usually rely on specific information-sharing schemes, perturbation structures, or system assumptions. Their generalizability across different federated learning architectures and against stronger adaptive attacks therefore remains limited.  

Beyond conventional attacks, Deng et al.\cite{DCL2026} reveal a new structural membership privacy risk in federated large language models (FedLLMs). They propose ProjRes, a passive membership inference attack based on projection residuals. ProjRes exploits the relationship between sample hidden representations and client-uploaded gradients. It can infer membership from only a single round of gradients. This makes it particularly effective for FedLLMs that use parameter-efficient fine-tuning methods such as Adapters and LoRA. ProjRes shows that the hidden representations of training samples can be indirectly identified through the subspace formed by gradients of trainable parameters. Existing defenses mainly rely on random noise, gradient clipping, or sparsification. However, weakening this type of attack often comes at the cost of model performance. Therefore, effectively disrupting the relationship between hidden representations and gradients while preserving model utility remains an important open problem in membership privacy protection for FedLLMs.

\begin{figure*}[ht]
	\centering
	\includegraphics[width=0.85\linewidth]{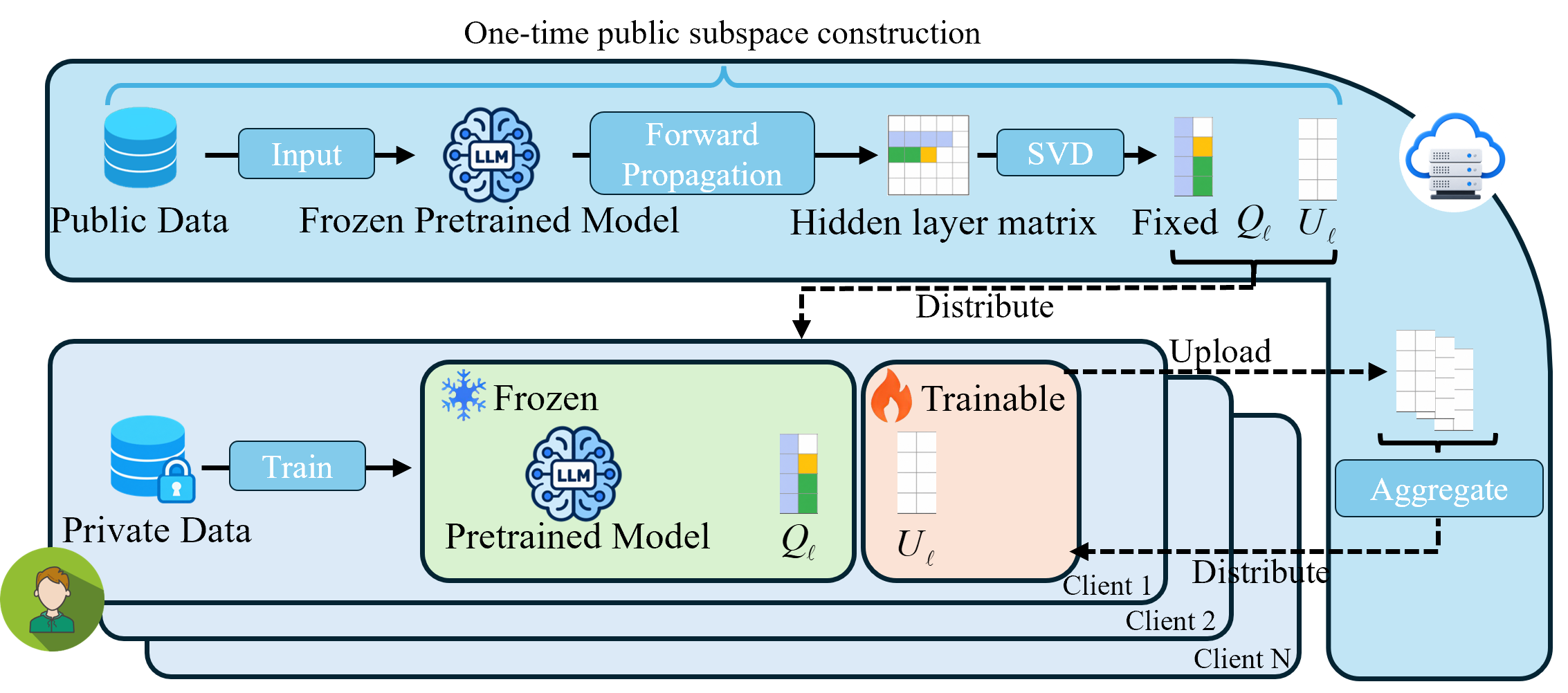}
	\caption{Overview of the proposed FISGuard.}
	\label{fig:fisguard_overview}
\end{figure*}

\section{FISGuard}

Figure~\ref{fig:fisguard_overview} provides an overview of FISGuard. The proposed method is not tied to a particular Adapter or LoRA design. It applies to any linear layer that is trained on the client, receives representations produced from private client samples, and exposes its update to the server. The method first identifies such layers from their common gradient structure, then builds a fixed input projection from public data, and finally rewrites each task-specific update as a fixed input projection followed by a trainable output mapping. Adapter and LoRA are two concrete implementations of this general construction.\\
\textbf{Step 1: Construct Fixed Input Subspaces and Reparameterize Target Linear Updates.}\\
Before federated training begins, the server performs a one-time model preparation step using only public data. The key idea is to remove trainable mappings that directly receive high-dimensional private representations, replace them with fixed public projections, and leave only the subsequent output mappings trainable on clients. In this way, clients can still adapt the model to their local tasks without learning or uploading the input-side mappings exploited by ProjRes.\\
\textbf{Step 1.1: }
The server first identifies trainable linear layers whose updates can directly expose private input representations and collects them in
$\mathcal{L}_{\mathrm{tar}}.$
Consider a linear layer
$y_\ell = W_\ell x_\ell + b_\ell,$
where \(x_\ell \in \mathbb{R}^{d_\ell}\) is the representation entering the layer and \(W_\ell \in \mathbb{R}^{m_\ell \times d_\ell}\) is its weight matrix. For a local batch of \(B\) samples, the weight gradient is
$
\nabla_{W_\ell}\mathcal{L}
=
\sum_{i=1}^{B}
\delta_{\ell,i} x_{\ell,i}^{\top},
$
where \(\delta_{\ell,i}\) is the error signal backpropagated to the layer output.

This gradient structure is important because each private representation \(x_{\ell,i}\) directly contributes to the uploaded update. As a result, the row space of the gradient is related to the space spanned by the private representations. ProjRes exploits exactly this relation: it measures how well a candidate representation lies in the gradient row space and uses the resulting projection residual for membership inference.

We therefore protect a linear layer when three conditions hold: (1) the layer is trained on the client, (2) it directly receives representations generated from private samples, and (3) its update is observable by the server. This criterion is independent of a specific fine-tuning architecture and can apply to attention projections, feed-forward linear layers, adaptation modules, and linear task heads.

For each target layer \(\ell\), the input dimension \(d_\ell\) is determined by the model architecture, while the server selects a smaller subspace dimension \(r_\ell\), with
$r_\ell < d_\ell.$\\
\textbf{Step 1.2:}
After identifying where protection is needed, the server must construct a replacement for the original trainable input-side mapping. Rather than choosing this mapping randomly, the server uses public data to identify the dominant representation directions of each target layer.

Specifically, the server prepares a public dataset
$\mathcal{D}_{\mathrm{pub}},$
which is independent of all client data. It freezes the pretrained model, performs forward passes on the public samples, and records the representation entering each target layer. For layer \(\ell\), the representation produced by the \(j\)-th public sample is denoted by
$x_{\ell,j}^{\mathrm{pub}}
\in
\mathbb{R}^{d_\ell}.$
The collected representations are stacked into
\[
X_\ell^{\mathrm{pub}}
=
\begin{bmatrix}
	\left(x_{\ell,1}^{\mathrm{pub}}\right)^{\top}\\
	\left(x_{\ell,2}^{\mathrm{pub}}\right)^{\top}\\
	\vdots\\
	\left(x_{\ell,N_{\mathrm{pub}}}^{\mathrm{pub}}\right)^{\top}
\end{bmatrix}
\in
\mathbb{R}^{N_{\mathrm{pub}} \times d_\ell}.
\]\\
At this stage, the server only collects representations. No client gradients are used and no model parameters are updated. The matrix \(X_\ell^{\mathrm{pub}}\) simply provides a public view of the representation space at the target layer and serves as the basis for constructing the fixed projection in the next step.\\
\textbf{Step 1.3:}
The server next extracts the dominant directions from the public representations. It applies singular value decomposition to
$X_\ell^{\mathrm{pub}}
=
P_\ell \Sigma_\ell V_\ell^{\top},$
and selects the right singular vectors associated with the largest \(r_\ell\) singular values. These directions are stacked to form
$Q_\ell
=
V_{\ell,1:r_\ell}^{\top}
\in
\mathbb{R}^{r_\ell \times d_\ell}.$
Because the selected directions are orthonormal,
$Q_\ell Q_\ell^{\top}
=
I_{r_\ell}.$
For any representation \(x_\ell\), the projection
$
z_\ell
=
Q_\ell x_\ell
\in
\mathbb{R}^{r_\ell}
$ keeps only its coordinates along these dominant public directions.
Thus, \(Q_\ell\) can be viewed as a fixed dimensionality-reduction mapping derived from public representations. It does not store individual public samples. Instead, it captures the main directions along which the pretrained representations vary. Compared with a random fixed projection, this construction is intended to preserve more of the information already encoded by the pretrained model while reducing the representation from \(d_\ell\) dimensions to \(r_\ell\) dimensions.\\
\textbf{Step 1.4: }
Once \(Q_\ell\) has been constructed, the server uses it to replace the trainable input-side mapping at the target layer. Instead of learning an unrestricted task-specific update \(\Delta W_\ell\), we parameterize it as
$
\Delta W_\ell
=
U_\ell Q_\ell,
$
where $U_\ell
\in
\mathbb{R}^{m_\ell \times r_\ell}$ is the only layer-specific mapping optimized on private client data.
The resulting transformation becomes
$
y_\ell
=
W_{0,\ell} x_\ell
+
U_\ell Q_\ell x_\ell,
$ where \(W_{0,\ell}\) is the frozen pretrained weight. Equivalently, the computation can be understood in two stages
$
z_\ell
=
Q_\ell x_\ell,
$
followed by
$
\Delta y_\ell
=
U_\ell z_\ell.
$
The roles of the two matrices are therefore separated clearly. The fixed matrix \(Q_\ell\) determines which directions of the high-dimensional input are retained, while the trainable matrix \(U_\ell\) learns how to use those retained directions for the downstream task. The server freezes both \(W_{0,\ell}\) and \(Q_\ell\), and only \(U_\ell\) is optimized on clients.

For an ordinary fully connected layer that does not already contain a low-rank adaptation structure, the server adds \(U_\ell Q_\ell\) as a task-specific branch on top of the frozen pretrained weight
$
y_\ell
=
W_{0,\ell} x_\ell
+
U_\ell Q_\ell x_\ell.
$
For a module that already contains an input projection followed by an output projection, the same principle is applied directly: the input-side mapping is replaced by the fixed public projection \(Q_\ell\), while the output-side mapping remains trainable.
To ensure that this additional branch does not change the pretrained model at initialization, the server sets
$
U_\ell^{0}
=
0.
$
Then
$
U_\ell^{0} Q_\ell x_\ell
=
0,
$
and therefore
$
y_\ell
=
W_{0,\ell} x_\ell.
$
The model thus starts from the original pretrained function and gradually learns task-specific corrections through \(U_\ell\).

\paragraph{Bottleneck Adapter.}

A standard Bottleneck Adapter computes
$
\Delta y_\ell
=
W_{\mathrm{up},\ell}
\sigma\!\left(W_{\mathrm{down},\ell}x_\ell\right).
$
\(W_{\mathrm{down},\ell}\) is the first matrix that acts on the high-dimensional representation \(x_\ell\). It determines how the original representation is compressed into the bottleneck space. The following matrix \(W_{\mathrm{up},\ell}\) then maps this low-dimensional representation back to the output space.
FISGuard uses the identity activation, \(\sigma(a)=a\), in the protected Adapter branch. We replace the trainable down projection with the fixed public projection,
$
W_{\mathrm{down},\ell}
=
Q_\ell,
$
and keep the up projection trainable,
$
W_{\mathrm{up},\ell}
=
U_\ell.
$
The Adapter branch therefore becomes
$
\Delta y_\ell
=
U_\ell Q_\ell x_\ell.
$
The public projection \(Q_\ell\) fixes the retained input directions. The trainable mapping \(U_\ell\) learns how to use the resulting low-dimensional representation. This yields the Fixed-Down Adapter instantiation of FISGuard.

\paragraph{LoRA.}

Standard LoRA writes the task-specific weight update as
$
\Delta W_\ell
=
B_\ell A_\ell,
$
which gives
$
\Delta y_\ell
=
B_\ell A_\ell x_\ell.
$
The factor \(A_\ell\) is the first matrix applied to the high-dimensional input and maps it into the low-rank space. The factor \(B_\ell\) then maps this low-rank representation to the output space. Thus, \(A_\ell\) and \(B_\ell\) play the same input-side and output-side roles as \(Q_\ell\) and \(U_\ell\), respectively.
We therefore set
$
A_\ell
=
Q_\ell,
B_\ell
=
U_\ell,
$
which yields
$
\Delta y_\ell
=
U_\ell Q_\ell x_\ell.
$
Again, the client no longer learns the input-side low-rank projection. The public projection \(Q_\ell\) remains fixed, while only \(U_\ell\) is optimized during federated training. This gives the LoRA instantiation of the same general reparameterization principle.

Both Bottleneck Adapter and LoRA therefore reduce to the common form
$
\Delta W_\ell
=
U_\ell Q_\ell.
$
They are two examples of the proposed method rather than its scope. More generally, the same construction can be applied whenever a task-specific linear update can be expressed as a fixed input projection followed by a trainable output mapping.\\
\textbf{Step 2: Distribute Global Trainable Parameters.}
After constructing the fixed input subspaces, the system enters multi-round federated training. At the beginning of each round, the server selects participating clients and sends the current trainable parameters
$
\Phi
=
\left\{U_\ell\mid \ell\in\mathcal{L}_{\mathrm{tar}}\right\}
\cup\Theta_{\mathrm{aux}},
$
where $\Theta_{\mathrm{aux}}$ denotes other task parameters that are not target linear weights. If any linear weight in this set also satisfies the conditions in Step 1.1, it must be included in $\mathcal{L}_{\mathrm{tar}}$ rather than left unprotected.
The base weights $W_{0,\ell}$ and fixed projections $Q_\ell$ are shared across clients and remain frozen throughout training. Each client initializes its local model with the received global parameters.\\
\textbf{Step 3: Perform Local Training in the Fixed Input Subspace.}
Each participating client trains on its private dataset $\mathcal{D}_k$. For every target layer $\ell$, the client computes
$
z_\ell=Q_\ell x_\ell,
\Delta y_\ell=U_\ell z_\ell,
$
which gives
$
y_\ell=W_{0,\ell}x_\ell+U_\ell Q_\ell x_\ell.
$
The client updates only $U_\ell$ and $\Theta_{\mathrm{aux}}$, while
\[
\nabla_{W_{0,\ell}}\mathcal{L}=0,
\qquad
\nabla_{Q_\ell}\mathcal{L}=0.
\]
Private data therefore no longer train an input matrix that acts directly on the original high-dimensional representation. They only train how the fixed low-dimensional representation should be mapped to the task output. The gradient of the output mapping is
$
\nabla_{U_\ell}\mathcal{L}
=
\sum_{i=1}^{B}\delta_{\ell,i}
\left(Q_\ell x_{\ell,i}\right)^{\top}.
$
The representation appearing in this gradient is the low-dimensional projected representation rather than the original high-dimensional input.\\
\textbf{Step 4: Upload Local Updates.}
After local training, client $k$ uploads only the updates of the output mappings and other trainable parameters
$
\Delta\Phi_k
=
\left\{\Delta U_{\ell,k}\mid \ell\in\mathcal{L}_{\mathrm{tar}}\right\}
\cup\Delta\Theta_{\mathrm{aux},k}.
$
Because $Q_\ell$ remains fixed, the client does not upload its gradient or parameter change. The input-side parameters that directly receive high-dimensional private representations are therefore absent from the server-visible update set. For Adapter, the client uploads only the up-projection update; for LoRA, it uploads only the output-side low-rank factor; other target linear layers upload their corresponding $U_\ell$ updates.\\
\textbf{Step 5: Aggregate Client Updates.}
The server receives the local updates and applies an existing federated aggregation rule. Under FedAvg, the aggregation weight of client $k$ is
$
p_k
=
\frac{|\mathcal{D}_k|}
{\sum_{j\in\mathcal{S}}|\mathcal{D}_j|},
$
where $\mathcal{S}$ is the set of participating clients in the current round. The server updates
$
\Phi
\leftarrow
\Phi+
\sum_{k\in\mathcal{S}}p_k\Delta\Phi_k.
$
The fixed projections $Q_\ell$ and base weights $W_{0,\ell}$ remain unchanged during aggregation. The server then starts the next round and repeats parameter distribution, local training, update upload, and aggregation until the training budget is exhausted or the model converges.

In summary, Step 1 is executed once before federated training, whereas Steps 2--5 are repeated in every round. The method does not change the basic federated optimization and aggregation workflow; it changes only how target linear layers are parameterized and which parameters clients train and upload.

\section{Experiments}\label{sec:experiments}
Experimental Environment. All experiments were conducted on a workstation running Windows 11, equipped with a 13th-generation Intel Core i5-13600KF CPU operating at 3.50 GHz, an NVIDIA GeForce RTX 3090 Ti GPU with 24 GB of VRAM, and 64 GB of system memory.

\subsection{Experimental Setup}
\textbf{Datasets.} 
We conduct experiments on three widely used NLP datasets: CoLA, SST-5, and IMDb. CoLA is a binary linguistic acceptability classification dataset consisting of English sentences labeled as grammatically acceptable or unacceptable. We use 8,551 samples for training and 1,043 samples for validation. SST-5, derived from the Stanford Sentiment Treebank, is a fine-grained sentiment classification dataset with five sentiment categories ranging from very negative to very positive. We use 8,544 training samples and 1,101 validation samples. IMDb is a binary sentiment classification dataset composed of movie reviews labeled as positive or negative. From its official 25,000-sample training set, we randomly split 22,500 samples for training and 2,500 samples for validation.

\textbf{Models.}
We consider two representative Transformer-based language models: BERT-Base and GPT-2 Large. BERT-Base is a bidirectional Transformer encoder with 12 layers, a hidden size of 768, and 12 attention heads, containing approximately 110M parameters. GPT-2 Large is an autoregressive Transformer language model with 36 layers, a hidden size of 1,280, and approximately 762M parameters.

\textbf{Baselines.}
We consider five representative defense methods as baselines: \textbf{DP-SGD}, \textbf{L2 Regularization}, \textbf{Min-Max}, \textbf{Dropout}, and \textbf{Label Smoothing}. 
DP-SGD limits the influence of individual training samples on model updates through gradient clipping and noise injection. 
L2 Regularization constrains the magnitude of PEFT parameters to mitigate overfitting. 
Min-Max jointly optimizes the target model and a membership inference adversary to explicitly reduce the distinguishability between members and non-members. 
Dropout randomly masks intermediate features during training to reduce feature co-adaptation and sample-specific memorization, while Label Smoothing softens the training targets to alleviate overconfident predictions on training samples. 
Collectively, these baselines mitigate membership information leakage from complementary perspectives, including gradient perturbation, regularization, and adversarial training, and serve as representative defenses for comparison with FISGuard.

\textbf{Hyperparameter Settings.}\label{sec:hy_setting}
We use a unified training configuration across all defense methods to ensure a fair comparison. Specifically, the federated learning system consists of $30$ clients, all of which participate in every communication round, and training proceeds for $50$ communication rounds. Each client performs one local update step per round with a local batch size of $16$. We use a learning rate of $2\times10^{-4}$ for the PEFT modules. The learning rate of the task-specific classification head is set to $5\times10^{-4}$ for CoLA and IMDB and $1\times10^{-3}$ for SST-5. Gradients are clipped to a global norm of $1.0$. In addition, we use a public reference set containing $2{,}000$ samples. For each defense baseline, we sweep its primary defense-specific hyperparameter while keeping all other training configurations unchanged. For DP-SGD, we fix the clipping norm to $C=1.0$ and $\delta=10^{-5}$, and vary the noise multiplier over $\sigma \in {0.1, 0.25, 0.5, 1, 2}$. For L2 Regularization, we sweep the regularization coefficient over $\lambda \in {10^{-4}, 10^{-3}, 10^{-2}, 10^{-1}}$. For Min-Max, we vary the adversarial loss weight over $\lambda \in {0.1, 0.3, 1, 3, 10}$, set the adversary learning rate to $3\times10^{-3}$, perform five adversary updates per mini-batch, use a hidden dimension of $128$, and sample $256$ held-out examples as non-member references. For Dropout, we consider dropout rates $p \in {0.05, 0.1, 0.2, 0.4}$, while for Label Smoothing, we vary the smoothing coefficient over $\alpha \in {0.02, 0.05, 0.1, 0.2, 0.4}$. For membership inference evaluation, we construct a balanced evaluation set containing $100$ member samples and $100$ non-member samples.

\begin{figure*}[htbp]
	\centering
	\includegraphics[width=\linewidth]
	{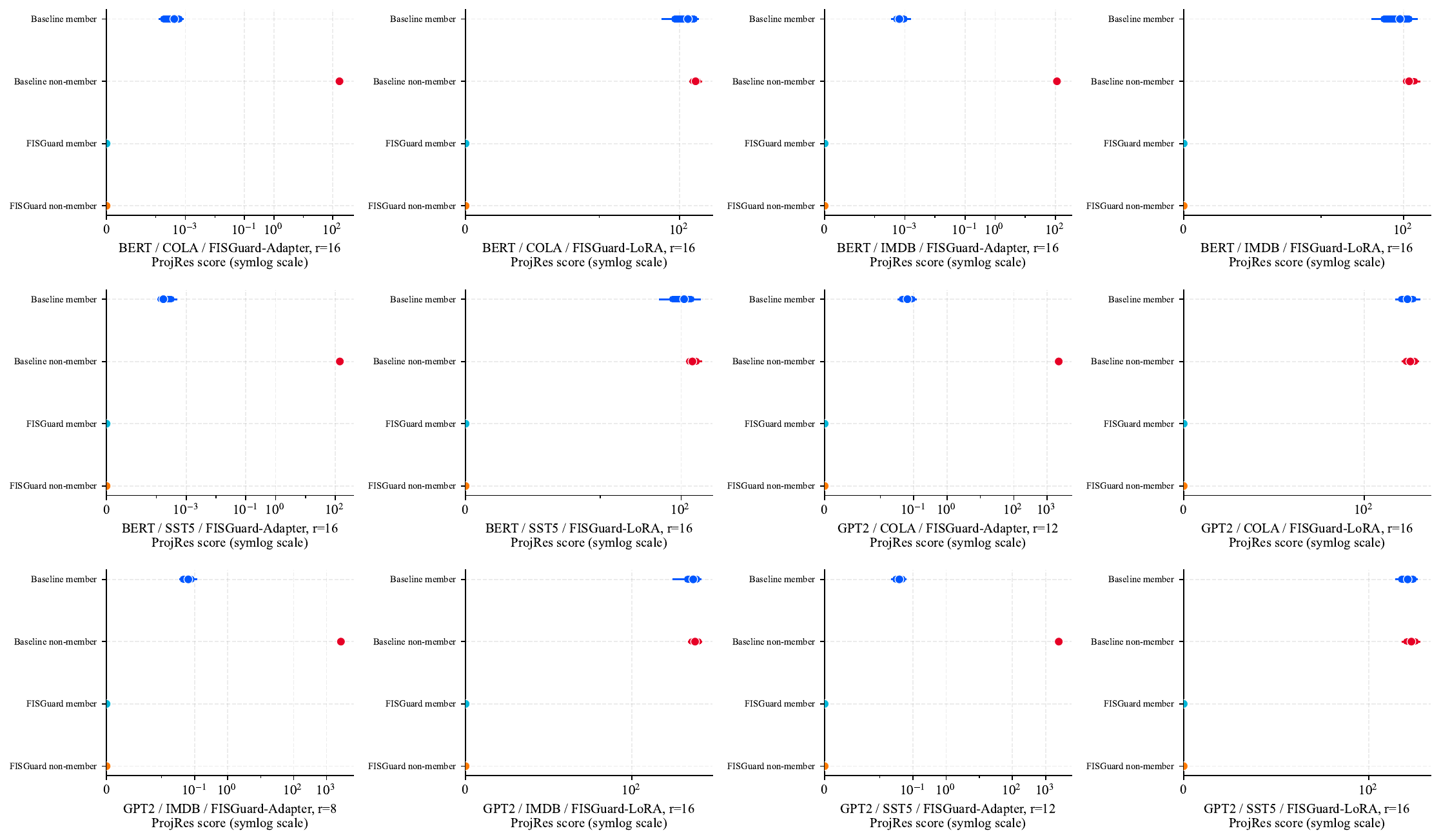}
	\caption{Projection residual score distributions of member and non-member samples with and without FISGuard.}
	\label{fig:member_nonmember_score}
\end{figure*}

\begin{figure}[htbp]
	\centering
	\includegraphics[width=\linewidth]
	{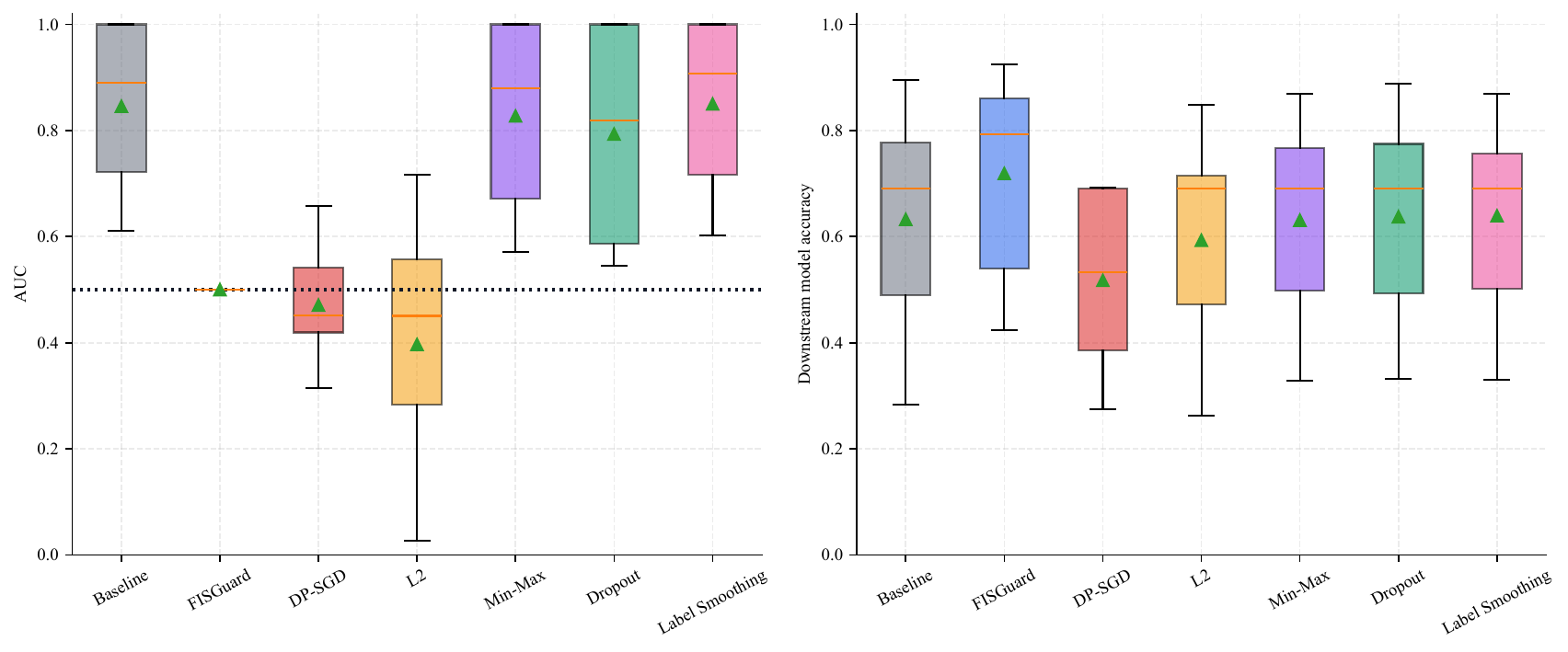}
	\caption{Overall comparison of ProjRes attack AUC and downstream model accuracy.}
	\label{fig:overall_performance}
\end{figure}

\begin{figure}[htbp]
	\centering
	\includegraphics[width=\linewidth]
	{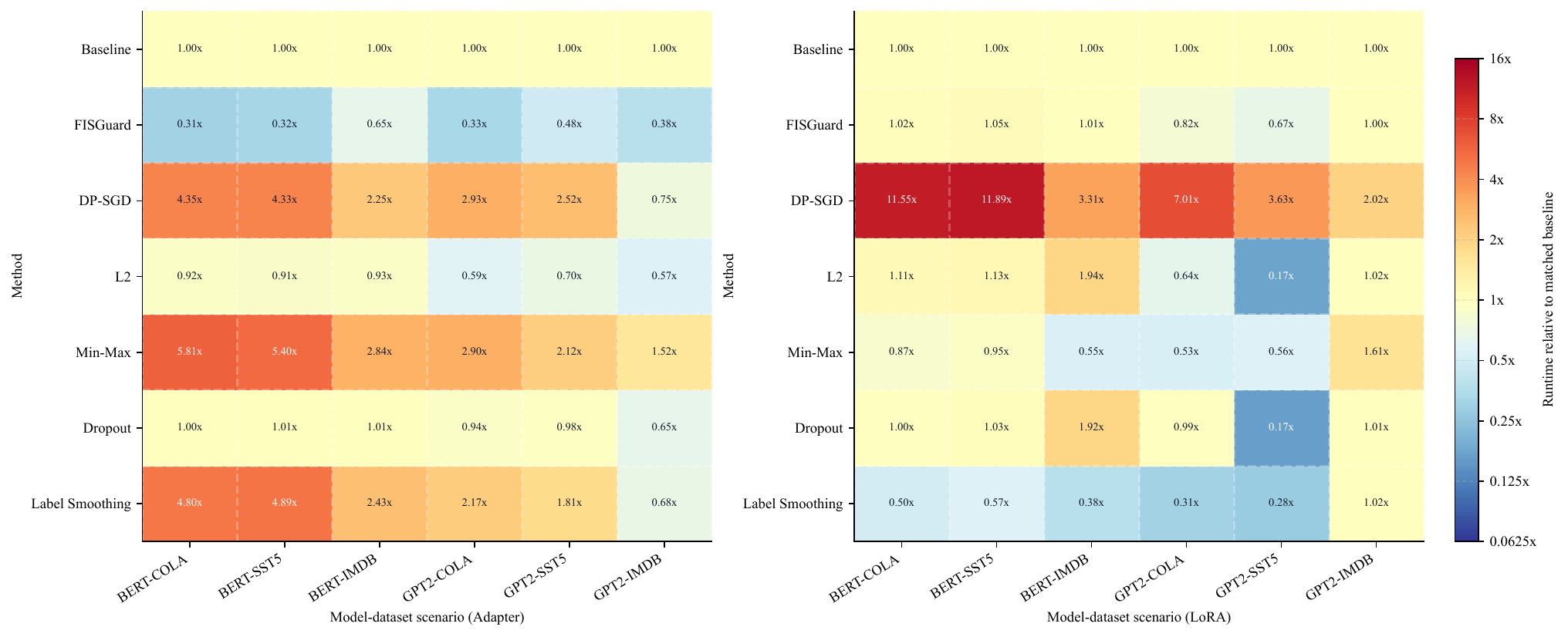}
	\caption{Training time comparison across different datasets under Adapter- and LoRA-based fine-tuning.}
	\label{fig:method_runtime_relative}
\end{figure}

\subsection{Defense Effectiveness and Model Utility}

We compare FISGuard against five representative defenses against membership inference attacks across three datasets, two LLMs, and two fine-tuning strategies, resulting in 12 experimental configurations. Figure~\ref{fig:overall_performance} summarizes their overall privacy and utility performance across these configurations. Detailed comparisons under different hyperparameter settings are provided in Appendix~\ref{app:detailed_results}.

\textbf{Privacy Protection.}
ProjRes performs membership inference using the projection residual score
$s(x)=\left\|h(x)-h(x)V_GV_G^{\top}\right\|_{1}$,
where $V_G$ denotes an orthonormal basis of the subspace induced by the observed gradient and $h(x)$ is the representation of candidate $x$. A smaller residual indicates that the candidate representation is better covered by the gradient subspace and is therefore more likely to be inferred as a member.

As shown in Fig.~\ref{fig:member_nonmember_score}, the undefended setting exhibits a clear separation between member and non-member scores. After applying FISGuard, the gap between the two distributions is substantially reduced and their overlap becomes much larger, indicating that the projection residual provides much weaker membership-discriminative information.

The AUC results in Fig.~\ref{fig:overall_performance} further confirm this observation. FISGuard produces an attack AUC distribution that is substantially more concentrated around 0.5 than those of the competing defenses. Across the 12 configurations, its attack AUC remains close to the random-guessing level in most cases, demonstrating consistent protection across different datasets, models, and fine-tuning strategies.

This behavior follows directly from the design of FISGuard. By freezing the input-side projection, the original high-dimensional private representations no longer directly contribute to the server-visible input-side gradient structure. Both member and non-member representations are instead mapped into the same fixed $r$-dimensional public subspace, which can be more sufficiently covered by the row space of the remaining observable gradients. As the coverage increases, the projection residuals of members and non-members become increasingly similar. When the gradient row space fully spans the projected space, both are completely covered, eliminating the residual difference exploited by ProjRes.

\textbf{Model Utility.}
The accuracy results in Fig.~\ref{fig:overall_performance} further show that the improved privacy protection of FISGuard does not incur substantial degradation in model utility. Its accuracy remains close to the undefended baseline and is generally more favorable than those of the competing defenses across the 12 configurations.

Unlike defenses that directly introduce noise or strong optimization constraints, FISGuard constructs its fixed input projection from public representations before private training. Specifically, singular value decomposition is applied to representations collected from public data, and the top-$r$ dominant representation directions are retained to form the fixed projection $Q$. These directions preserve the principal information encoded in the pretrained representation space. During private training, $Q$ remains frozen while only the output-side mapping $U$ is optimized, allowing the model to adapt the retained representations to downstream tasks while restricting the gradient information exploitable by ProjRes.

Taken together, these results show that FISGuard effectively suppresses the membership-discriminative signal exploited by ProjRes while preserving competitive downstream task performance, achieving a favorable privacy--utility trade-off across different datasets, models, and fine-tuning strategies.

\subsection{Computation Overhead Analysis}

As shown in Fig.\ref{fig:method_runtime_relative}, across different datasets, FISGuard introduces only limited computational overhead. Under Adapter-based fine-tuning, FISGuard achieves the lowest overall training time among the evaluated defense methods, with runtime remaining close to the undefended baseline. Under LoRA-based fine-tuning, its training time also stays close to the baseline, indicating that the additional privacy protection incurs only modest computational cost. This efficiency is mainly because the public-data-derived input projection is constructed before private training and remains frozen throughout fine-tuning, so no gradient computation or parameter update is required for this component. Moreover, FISGuard does not introduce additional per-step procedures such as gradient perturbation or adversarial optimization. Overall, the runtime results show that FISGuard provides effective privacy protection without substantially increasing the computational cost of parameter-efficient fine-tuning.

\section{Security Analysis}
\label{sec:security-analysis}

This section analyzes how fixed public-subspace reparameterization
affects ProjRes. The guarantee established below is specific to the
projection-residual score used by ProjRes. It is neither a
differential-privacy guarantee nor a universal guarantee against all
membership-inference attacks.

\subsection{Setting}
\label{sec:setting}

Consider a protected linear layer whose task-specific update is
reparameterized as \(\Delta W = UQ\), so that the output corresponding
to sample \(i\) is $y_i = W_0x_i + UQx_i$.
\(x_i \in \mathbb{R}^{d}\) is the high-dimensional
representation entering the layer, \(W_0\) is the frozen base weight,
\(Q \in \mathbb{R}^{r \times d}\) is a fixed input projection
constructed from public data, and \(U \in \mathbb{R}^{m \times r}\)
is the output mapping trained and uploaded by the client. We assume
that \(r<d\).
Define the projected representation as
\(z_i := Qx_i \in \mathbb{R}^{r}\).
Let \(\delta_i \in \mathbb{R}^{m}\) denote the backpropagated error
signal associated with sample \(i\). Absorbing any nonzero
batch-normalization factor into \(\delta_i\), the gradient of the loss
with respect to \(U\) is
$
G_U
:=
\nabla_U \mathcal{L}
=
\sum_{i=1}^{B}\delta_i z_i^{\top}
\in
\mathbb{R}^{m \times r}.
$
Because \(Q\) is excluded from the client-side optimization variables,
the client uploads neither a gradient nor a parameter difference for
\(Q\). The server therefore observes \(G_U\), rather than an
input-side update directly constructed from the original
high-dimensional representations \(x_i\).
For a candidate sample \(x\), let \(z:=Qx\). The ProjRes score
computed from the observed update \(G_U\) is
$
s_{\mathrm{PR}}(x;G_U)
:=
\left\|
z-\Pi_{\operatorname{row}(G_U)}z
\right\|_p,
$
where \(\Pi_{\operatorname{row}(G_U)}\) denotes the orthogonal
projection onto the row space of \(G_U\).

\subsection{Invalidation of the ProjRes Score}
\label{sec:projres-invalidation}

\begin{theorem}[ProjRes invalidation under a full-rank fixed subspace]
	\label{thm:projres-invalidation}
	Suppose that the output-side update observed by the server satisfies
	\(\operatorname{rank}(G_U)=r\). Then, for every candidate sample \(x\),
	regardless of whether \(x\) belongs to the client's training data,
	$s_{\mathrm{PR}}(x;G_U)=0$.
	Consequently, ProjRes cannot rank or classify members and non-members
	using the projection-residual score.
\end{theorem}

\begin{proof}
	Because \(G_U \in \mathbb{R}^{m \times r}\), its row space is a
	subspace of \(\mathbb{R}^{r}\), namely,
	$\operatorname{row}(G_U)\subseteq\mathbb{R}^{r}.$
	Under the condition \(\operatorname{rank}(G_U)=r\), the dimension of
	the row space equals the dimension of the entire projected
	representation space. Therefore,$\operatorname{row}(G_U)=\mathbb{R}^{r}.$
	Orthogonal projection onto all of \(\mathbb{R}^{r}\) is the identity
	map, so \(\Pi_{\operatorname{row}(G_U)}=I_r\).
	For any candidate sample \(x\), the projected representation \(z=Qx\)
	belongs to \(\mathbb{R}^{r}\). 
	$
	\Pi_{\operatorname{row}(G_U)}z
	=
	I_rz
	=
	z.
	$
	It follows that
	$
	s_{\mathrm{PR}}(x;G_U)
	=
	\left\|z-z\right\|_p
	=
	0.
	$
	This conclusion is independent of the membership status of \(x\).
	Thus, the ProjRes score collapses to a constant and contains no
	membership-discriminating information.
\end{proof}

\begin{corollary}[Zero membership advantage for projection-residual attacks]
	\label{cor:zero-membership-advantage}
	Let \(M\in\{0,1\}\) denote the membership variable, where \(M=1\)
	represents a member and \(M=0\) represents a non-member. Define the
	threshold-based ProjRes advantage as $\operatorname{Adv}_{\mathrm{PR}}
	:=
	\sup_{\tau\in\mathbb{R}}
	\Bigg|\Pr\!\left[
	s_{\mathrm{PR}}(x;G_U)\leq\tau
	\,\middle|\,
	M=1
	\right]
	\\-
	\Pr\!\left[
	s_{\mathrm{PR}}(x;G_U)\leq\tau
	\,\middle|\,
	M=0
	\right]
	\Bigg|.$
	Under the conditions of Theorem~\ref{thm:projres-invalidation},
	\(\operatorname{Adv}_{\mathrm{PR}}=0\). Under the standard convention
	that tied scores receive half credit, the corresponding ROC AUC is
	\(\operatorname{AUC}=1/2\).
\end{corollary}

\begin{proof}
	By Theorem~\ref{thm:projres-invalidation},
	\(s_{\mathrm{PR}}(x;G_U)=0\) for both members and non-members.
	Therefore, for every threshold \(\tau\),
	\[
	\Pr\!\left[
	s_{\mathrm{PR}}(x;G_U)\leq\tau
	\,\middle|\,
	M=1
	\right]
	=
	\Pr\!\left[
	s_{\mathrm{PR}}(x;G_U)\leq\tau
	\,\middle|\,
	M=0
	\right].
	\]
	The two conditional score distributions are identical, so their
	probability difference is zero for every \(\tau\). Consequently,
	\(\operatorname{Adv}_{\mathrm{PR}}=0\).
	Moreover, all candidates receive the same score.  ProjRes
	produces no nontrivial ranking between members and non-members. With
	the standard treatment of tied scores, its ROC AUC is \(1/2\).
\end{proof}

\subsection{A Sufficient Condition for Full Rank}
\label{sec:full-rank-condition}
Define the projected-representation matrix \(Z\) and the error-signal
matrix \(D\) as
\[
Z :=
\begin{bmatrix}
	z_1 & z_2 & \cdots & z_B
\end{bmatrix}
\in \mathbb{R}^{r \times B},
\]
\[
D :=
\begin{bmatrix}
	\delta_1 & \delta_2 & \cdots & \delta_B
\end{bmatrix}
\in \mathbb{R}^{m \times B}.
\]
The output-side gradient can then be written as \(G_U=DZ^{\top}\).

\begin{proposition}[Sufficient condition for a full-rank output update]
	\label{prop:sufficient-full-rank}
	Suppose that \(\operatorname{rank}(Z)=r\) and
	$
	\ker(D)\cap\operatorname{col}(Z^{\top})=\{0\}.
	$
	Then \(\operatorname{rank}(G_U)=r\).
\end{proposition}

\begin{proof}
	Consider any vector \(a\in\mathbb{R}^{r}\) satisfying
	\(DZ^{\top}a=0\). It follows that \(Z^{\top}a\in\ker(D)\).
	By construction, \(Z^{\top}a\in\operatorname{col}(Z^{\top})\).
	Therefore,
	$
	Z^{\top}a
	\in
	\ker(D)\cap\operatorname{col}(Z^{\top}).
	$
	The assumed intersection condition implies that \(Z^{\top}a=0\).
	Because \(\operatorname{rank}(Z)=r\), the matrix \(Z^{\top}\) has full
	column rank and therefore has a trivial null space. \(a=0\).
	Thus, the linear map
	\(DZ^{\top}\colon\mathbb{R}^{r}\to\mathbb{R}^{m}\) is injective.
	Consequently,
	$
	\operatorname{rank}(DZ^{\top})=r.
	$
	Since \(G_U=DZ^{\top}\), we obtain
	\(\operatorname{rank}(G_U)=r\).
\end{proof}

\subsection{Interpretation and Design Condition}
\label{sec:design-condition}

Proposition~\ref{prop:sufficient-full-rank} has a direct geometric
interpretation. The projected batch representations must span all
\(r\) directions of the fixed public subspace, while the
backpropagated error signals must not eliminate any nonzero direction
in that span.
Because
$
\operatorname{rank}(G_U)
\leq
\min\{m,r,B\},
$
the full-rank condition \(\operatorname{rank}(G_U)=r\) requires
\(r\leq m\) and \(r\leq B\).
For sequence models, \(B\) may be replaced by the effective number of
linearly independent representations contained in the observed
update, denoted by \(B_{\mathrm{eff}}\). A practical dimensional
condition is therefore
$
r\leq\min\{m,B_{\mathrm{eff}}\}.
$
This dimensional condition is necessary for full rank but is not, by
itself, sufficient. The spanning and non-cancellation conditions in
Proposition~\ref{prop:sufficient-full-rank} must also hold.


\section{Conclusion}

In this paper, we presented FISGuard, a lightweight defense against the ProjRes membership inference attack in federated parameter-efficient fine-tuning. FISGuard constructs a fixed low-dimensional input subspace from independent public data and restricts trainable updates to operate within this subspace, thereby reducing the projection-residual discrepancy between members and non-members while preserving the information required for downstream tasks. 

We evaluated FISGuard across three NLP datasets, two large language models, and two fine-tuning strategies, and compared it with five representative membership-inference defenses. The results show that FISGuard reduces the AUC of the ProjRes attack to near the random-guessing level in most settings, while maintaining downstream task performance close to that of the undefended model and introducing only limited computational overhead. Our analysis further characterizes the conditions under which the server-observable update span fully covers the fixed representation subspace, explaining why the projection-residual signal exploited by ProjRes can be substantially weakened. Overall, FISGuard provides an effective privacy--utility trade-off against projection-residual-based membership inference without requiring heavy gradient perturbation or regularization.

\bibliographystyle{plainurl}
\bibliography{references}

@inproceedings{HLL2025,
	title={Towards Label-Only membership inference attack against pre-trained large language models},
	author={He, Yu and Li, Boheng and Liu, Liu and Ba, Zhongjie and Dong, Wei and Li, Yiming and Qin, Zhan and Ren, Kui and Chen, Chun},
	booktitle={34th USENIX Security Symposium (USENIX Security 25)},
	pages={1609--1628},
	year={2025}
}

@article{DCL2026,
	title={Toward Efficient Membership Inference Attacks against Federated Large Language Models: A Projection Residual Approach},
	author={Deng, Guilin and Chen, Silong and Luo, Yuchuan and Liu, Yi and Wang, Songlei and Cai, Zhiping and Liu, Lin and Jia, Xiaohua and Fu, Shaojing},
	journal={arXiv preprint arXiv:2604.21197},
	year={2026}
}

@inproceedings{SSS2017,
	title={Membership inference attacks against machine learning models},
	author={Shokri, Reza and Stronati, Marco and Song, Congzheng and Shmatikov, Vitaly},
	booktitle={2017 IEEE symposium on security and privacy (SP)},
	pages={3--18},
	year={2017},
	organization={IEEE}
}

@inproceedings{DLZ2025,
	title={Privacy in fine-tuning large language models: Attacks, defenses, and future directions},
	author={Du, Hao and Liu, Shang and Zheng, Lele and Cao, Yang and Nakamura, Atsuyoshi and Chen, Lei},
	booktitle={Pacific-Asia Conference on Knowledge Discovery and Data Mining},
	pages={326--344},
	year={2025},
	organization={Springer}
}

@article{HSS2022,
	title={Membership inference attacks on machine learning: A survey},
	author={Hu, Hongsheng and Salcic, Zoran and Sun, Lichao and Dobbie, Gillian and Yu, Philip S and Zhang, Xuyun},
	journal={ACM Computing Surveys (CSUR)},
	volume={54},
	number={11s},
	pages={1--37},
	year={2022},
	publisher={ACM New York, NY}
}

@inproceedings{MIH2021,
	title={Privacy regularization: Joint privacy-utility optimization in LanguageModels},
	author={Mireshghallah, Fatemehsadat and Inan, Huseyin and Hasegawa, Marcello and R{\"u}hle, Victor and Berg-Kirkpatrick, Taylor and Sim, Robert},
	booktitle={Proceedings of the 2021 Conference of the North American Chapter of the Association for Computational Linguistics: Human Language Technologies},
	pages={3799--3807},
	year={2021}
}

@incollection{PL2002,
	title={Early stopping-but when?},
	author={Prechelt, Lutz},
	booktitle={Neural Networks: Tricks of the trade},
	pages={55--69},
	year={2002},
	publisher={Springer}
}

@article{CWH2020,
	title={Understanding gradient clipping in private sgd: A geometric perspective},
	author={Chen, Xiangyi and Wu, Steven Z and Hong, Mingyi},
	journal={Advances in neural information processing systems},
	volume={33},
	pages={13773--13782},
	year={2020}
}

@inproceedings{DC2006,
	title={Differential privacy},
	author={Dwork, Cynthia},
	booktitle={International colloquium on automata, languages, and programming},
	pages={1--12},
	year={2006},
	organization={Springer}
}

@article{BIK2016,
	title={Practical secure aggregation for federated learning on user-held data},
	author={Bonawitz, Keith and Ivanov, Vladimir and Kreuter, Ben and Marcedone, Antonio and McMahan, H Brendan and Patel, Sarvar and Ramage, Daniel and Segal, Aaron and Seth, Karn},
	journal={arXiv preprint arXiv:1611.04482},
	year={2016}
}

@inproceedings{KQL2024,
	title={Federatedscope-llm: A comprehensive package for fine-tuning large language models in federated learning},
	author={Kuang, Weirui and Qian, Bingchen and Li, Zitao and Chen, Daoyuan and Gao, Dawei and Pan, Xuchen and Xie, Yuexiang and Li, Yaliang and Ding, Bolin and Zhou, Jingren},
	booktitle={Proceedings of the 30th ACM SIGKDD Conference on Knowledge Discovery and Data Mining},
	pages={5260--5271},
	year={2024}
}

@inproceedings{YZL2020,
	title={Privacy-preserving in defending against membership inference attacks},
	author={Ying, Zuobin and Zhang, Yun and Liu, Ximeng},
	booktitle={Proceedings of the 2020 workshop on privacy-preserving machine learning in practice},
	pages={61--63},
	year={2020}
}

@inproceedings{NSH2018,
	title={Machine learning with membership privacy using adversarial regularization},
	author={Nasr, Milad and Shokri, Reza and Houmansadr, Amir},
	booktitle={Proceedings of the 2018 ACM SIGSAC conference on computer and communications security},
	pages={634--646},
	year={2018}
}

@inproceedings{LLR2021,
	title={Membership inference attacks and defenses in classification models},
	author={Li, Jiacheng and Li, Ninghui and Ribeiro, Bruno},
	booktitle={Proceedings of the Eleventh ACM Conference on Data and Application Security and Privacy},
	pages={5--16},
	year={2021}
}

@article{ASH2025,
	title={MemberShield: A framework for federated learning with membership privacy},
	author={Ahmed, Faisal and Sanchez, David and Haddi, Zouhair and Domingo-Ferrer, Josep},
	journal={Neural Networks},
	volume={181},
	pages={106768},
	year={2025},
	publisher={Elsevier}
}

@article{JE2022,
	title={MIAShield: Defending membership inference attacks via preemptive exclusion of members},
	author={Jarin, Ismat and Eshete, Birhanu},
	journal={arXiv preprint arXiv:2203.00915},
	year={2022}
}

@article{HLL2022,
	title={Defending against membership inference attacks with high utility by GAN},
	author={Hu, Li and Li, Jin and Lin, Guanbiao and Peng, Shiyu and Zhang, Zhenxin and Zhang, Yingying and Dong, Changyu},
	journal={IEEE Transactions on Dependable and Secure Computing},
	volume={20},
	number={3},
	pages={2144--2157},
	year={2022},
	publisher={IEEE}
}

@article{BLZ2026,
	title={United We Defend: Collaborative Membership Inference Defenses in Federated Learning},
	author={Bai, Li and Liu, Junxu and Zhang, Sen and Zhang, Xinwei and Ye, Qingqing and Hu, Haibo},
	journal={arXiv preprint arXiv:2601.06866},
	year={2026}
}

@article{NHD2020,
	title={Local and central differential privacy for robustness and privacy in federated learning},
	author={Naseri, Mohammad and Hayes, Jamie and De Cristofaro, Emiliano},
	journal={arXiv preprint arXiv:2009.03561},
	year={2020}
}

@article{HWL2024,
	title={Privacy risks of federated knowledge graph embedding: New membership inference attacks and personalized differential privacy defense},
	author={Hu, Yuke and Wang, Yang and Lou, Jian and Liang, Wei and Wu, Ruofan and Wang, Weiqiang and Li, Xiaochen and Liu, Jinfei and Qin, Zhan},
	journal={IEEE Transactions on Dependable and Secure Computing},
	volume={22},
	number={3},
	pages={2788--2805},
	year={2024},
	publisher={IEEE}
}

@article{ZLS2025,
	title={ADPF: Anti-inference differentially private protocol for federated learning},
	author={Zhao, Zirun and Lin, Zhaowen and Sun, Yi},
	journal={Computer Networks},
	volume={261},
	pages={111130},
	year={2025},
	publisher={Elsevier}
}

@article{SWS2025,
	title={Towards the flatter landscape and better generalization in federated learning under client-level differential privacy},
	author={Shi, Yifan and Wei, Kang and Shen, Li and Liu, Yingqi and Wang, Xueqian and Yuan, Bo and Tao, Dacheng},
	journal={IEEE Transactions on Pattern Analysis and Machine Intelligence},
	year={2025},
	publisher={IEEE}
}

@article{YCW2025,
	title={DP-FedPUAC: Federated learning with differential privacy via adaptive gradient clipping and local iteration optimization},
	author={Yuan, Jiangyong and Chen, Yong and Wang, Zheyi and Wang, ChenXiao and Hu, Xiaoyue and Zeng, Zihao},
	journal={Information Sciences},
	pages={122981},
	year={2025},
	publisher={Elsevier}
}

@inproceedings{QWH2024,
	title={Towards the robustness of differentially private federated learning},
	author={Qi, Tao and Wang, Huili and Huang, Yongfeng},
	booktitle={Proceedings of the AAAI Conference on Artificial Intelligence},
	volume={38},
	number={18},
	pages={19911--19919},
	year={2024}
}

@inproceedings{AKR2025,
	title={Little is enough: Boosting privacy by sharing only hard labels in federated semi-supervised learning},
	author={Abourayya, Amr and Kleesiek, Jens and Rao, Kanishka and Ayday, Erman and Rao, Bharat and Webb, Geoffrey I and Kamp, Michael},
	booktitle={Proceedings of the AAAI conference on artificial intelligence},
	volume={39},
	number={15},
	pages={15293--15301},
	year={2025}
}

@inproceedings{YYH2023,
	title={Fortifying federated learning against membership inference attacks via client-level input perturbation},
	author={Yang, Yuchen and Yuan, Haolin and Hui, Bo and Gong, Neil and Fendley, Neil and Burlina, Philippe and Cao, Yinzhi},
	booktitle={2023 53rd Annual IEEE/IFIP International Conference on Dependable Systems and Networks (DSN)},
	pages={288--301},
	year={2023},
	organization={IEEE}
}

@article{SMX2024,
	title={MemDefense: defending against membership inference attacks in IoT-based federated learning via pruning perturbations},
	author={Shen, Meng and Meng, Jin and Xu, Ke and Yu, Shui and Zhu, Liehuang},
	journal={IEEE Transactions on Big Data},
	volume={11},
	number={5},
	pages={2135--2147},
	year={2024},
	publisher={IEEE}
}

@inproceedings{YFF2022,
	title={An accuracy-lossless perturbation method for defending privacy attacks in federated learning},
	author={Yang, Xue and Feng, Yan and Fang, Weijun and Shao, Jun and Tang, Xiaohu and Xia, Shu-Tao and Lu, Rongxing},
	booktitle={Proceedings of the ACM Web Conference 2022},
	pages={732--742},
	year={2022}
}

@inproceedings{HCB2021,
	title={Scaling federated learning for fine-tuning of large language models},
	author={Hilmkil, Agrin and Callh, Sebastian and Barbieri, Matteo and S{\"u}tfeld, Leon Ren{\'e} and Zec, Edvin Listo and Mogren, Olof},
	booktitle={International Conference on Applications of Natural Language to Information Systems},
	pages={15--23},
	year={2021},
	organization={Springer}
}

@article{LM2020,
	title={Federated pretraining and fine tuning of bert using clinical notes from multiple silos},
	author={Liu, Dianbo and Miller, Tim},
	journal={arXiv preprint arXiv:2002.08562},
	year={2020}
}

\appendices
\section{Detailed Results}
\label{app:detailed_results}

While the main text provides an aggregate comparison across the 12 experimental configurations, this section presents detailed results under the different hyperparameter settings specified in Section~\ref{sec:hy_setting}. Specifically, we report the ProjRes attack AUC and the corresponding downstream model accuracy as the defense-specific hyperparameters vary across all combinations of three datasets, two LLMs, and two fine-tuning strategies. For FISGuard, the defense parameter is the rank $r$, which determines the dimensionality of the fixed public subspace constructed from the top-$r$ dominant directions of public representations. These detailed results provide a more fine-grained view of how different defense parameters affect privacy protection and model utility.

\textbf{Attack AUC.} 
Figure~\ref{fig:auc_all_scenarios} presents the ProjRes attack AUC under the complete set of evaluated hyperparameter settings. The results provide a detailed comparison of how the effectiveness of each defense changes with its defense-specific parameter across different datasets, models, and fine-tuning strategies. For conventional defenses, varying the amount of noise, regularization strength, adversarial penalty, dropout probability, or label-smoothing strength changes the extent to which the original optimization process is perturbed, which in turn affects their ability to suppress the membership signal exploited by ProjRes.

For FISGuard, varying $r$ directly changes the dimensionality of the fixed subspace in which the remaining observable gradients are formed. Across the evaluated settings, FISGuard keeps the ProjRes AUC close to the random-guessing level in most configurations, further supporting the aggregate results reported in the main text. The effect of $r$ also admits a direct geometric interpretation. A smaller $r$ compresses both member and non-member representations into a lower-dimensional shared space, making a larger portion of this space easier to cover with the row space of the observable gradient. As this coverage increases, member and non-member representations share increasingly overlapping covered components, and the difference between their projection residuals correspondingly decreases. In the limiting case where the gradient row space spans the entire $r$-dimensional projected space, both types of candidates are fully covered and the ProjRes residual becomes identical regardless of membership. Conversely, increasing $r$ retains a larger representation space but also increases the dimensionality that must be covered by the observable gradient span. 

\textbf{Model Accuracy.} 
Figure~\ref{fig:model_accuracy_all_scenarios} reports the downstream model accuracy corresponding to the same hyperparameter sweeps. For several existing defenses, changing the defense strength directly alters the degree of perturbation or regularization applied during training. Stronger perturbation can suppress membership leakage more effectively, but may simultaneously interfere with normal optimization and reduce downstream task performance.

The role of the hyperparameter $r$ in FISGuard is different. Rather than controlling the magnitude of noise or gradient perturbation, $r$ determines how many dominant representation directions extracted from public data are retained in the fixed input subspace. The projection matrix $Q$ is constructed from the top-$r$ right singular vectors of the public representation matrix, such that the retained dimensions correspond to the principal directions of variation in the pretrained representation space. Increasing $r$ preserves a larger portion of the original representation space and provides the trainable output-side mapping with richer information for downstream adaptation, whereas reducing $r$ imposes stronger dimensionality reduction. Therefore, $r$ directly governs the balance between retaining task-relevant representation information and reducing the dimensionality of the observable space exploited by ProjRes. The accuracy results across different values of $r$ empirically illustrate how this representation-space constraint affects downstream model utility.

\begin{figure*}[htbp]
	\centering
	\includegraphics[width=\textwidth]
	{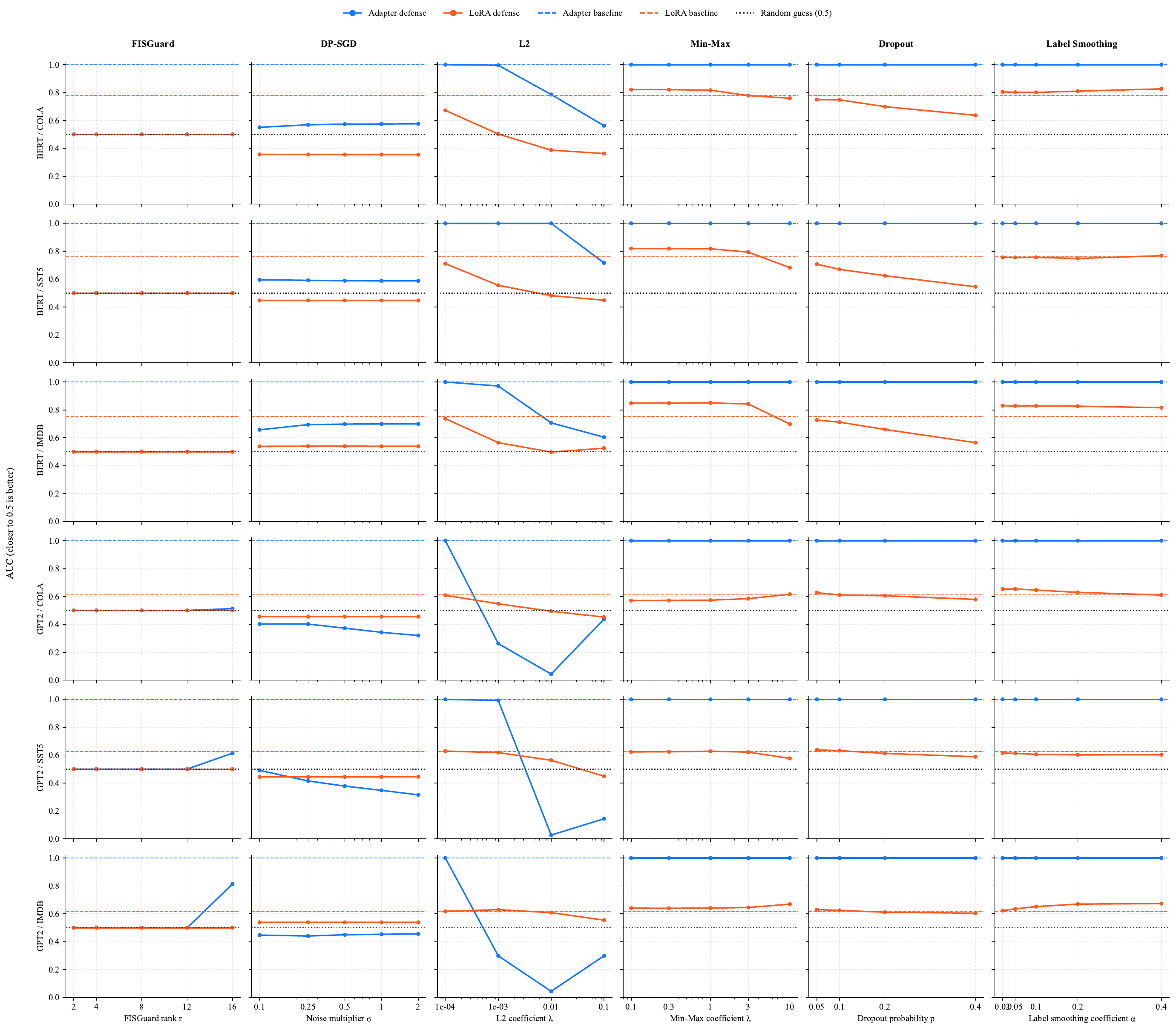}
	\caption{Attack AUC of ProjRes under different defense methods across three datasets, two LLMs, and two fine-tuning strategies.}
	\label{fig:auc_all_scenarios}
\end{figure*}

\begin{figure*}[htbp]
	\centering
	\includegraphics[width=\textwidth]
	{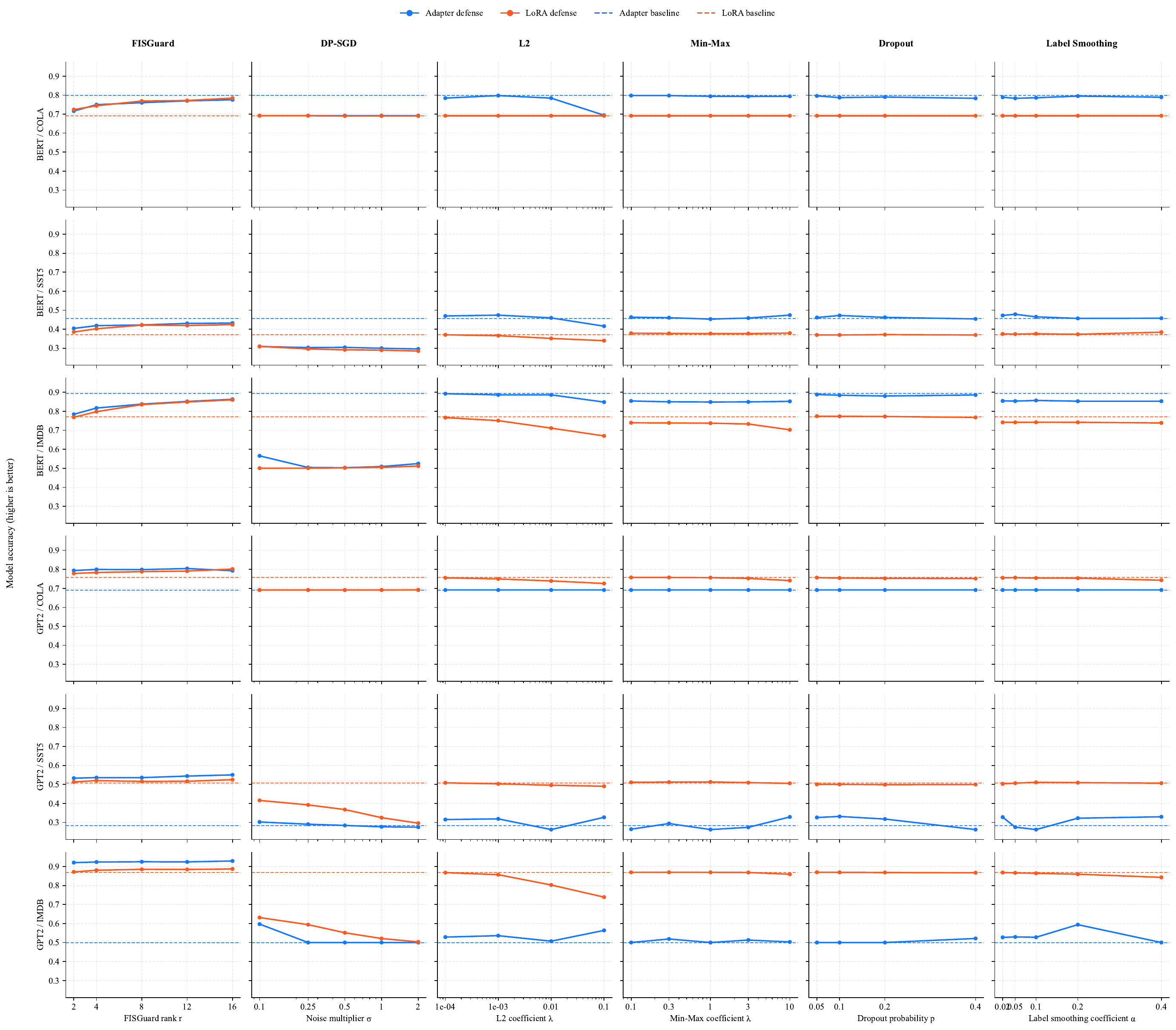}
	\caption{Model accuracy under different defense methods across three datasets, two LLMs, and two fine-tuning strategies.}
	\label{fig:model_accuracy_all_scenarios}
\end{figure*}



%

\end{document}